\documentclass[a4paper,UKenglish,modulo]{lipics-v2019}
\usepackage{algorithm}
\usepackage{graphicx}
\usepackage[noend]{algpseudocode}
\usepackage{amssymb}

\newcommand{\E}{\ensuremath{\mathbb{E}}}
\newcommand{\N}{\ensuremath{\mathbb{N}}}

\graphicspath{{./figures/}}
\def\BibTeX{{\rm B\kern-.05em{\sc i\kern-.025em b}\kern-.08emT\kern-.1667em\lower.7ex\hbox{E}\kern-.125emX}}
    
\nolinenumbers 
\hideLIPIcs  

\title{Fast Distributed Algorithms for LP-Type Problems of Bounded Dimension}

\author{Kristian Hinnenthal}{Paderborn University, Germany}{krijan@mail.upb.de}{}{}
\author{Christian Scheideler}{Paderborn University, Germany}{scheidel@mail.upb.de}{}{}
\author{Martijn Struijs}{TU Eindhoven, The Netherlands}{m.a.c.struijs@tue.nl}{}{}

\authorrunning{K. Hinnenthal, C. Scheideler, and M. Struijs} 

\relatedversion{} 
\funding{}
\acknowledgements{}

\Copyright{Kristian Hinnenthal, Christian Scheideler, and Martijn Struijs}



\begin{document}

%
\maketitle

\begin{abstract}
	In this paper we present various distributed algorithms for LP-type problems in the well-known gossip model. LP-type problems include many important classes of problems such as (integer) linear programming, geometric problems like smallest enclosing ball and polytope distance, and set problems like hitting set and set cover. In the gossip model, a node can only push information to or pull information from nodes chosen uniformly at random. Protocols for the gossip model are usually very practical due to their fast convergence, their simplicity, and their stability under stress and disruptions. Our algorithms are very efficient (logarithmic rounds or better with just polylogarithmic communication work per node per round) whenever the combinatorial dimension of the given LP-type problem is constant, even if the size of the given LP-type problem is polynomially large in the number of nodes.
\end{abstract}

\section{Introduction}

\subsection{LP-type problems}

LP-type problems were defined by Sharir and Welzl \cite{SW92} as problems characterized by a tuple $(H,f)$ where $H$ is a finite set and $f:2^H \rightarrow T$ is a function that maps subsets from $H$ to values in a totally ordered set $(T, \le)$ containing $\infty$. The function $f$ is required to satisfy two conditions:
\begin{itemize}
\item {\bf Monotonicity:}
For all sets $F \subseteq G \subseteq H$, $f(F) \le f(G) \le f(H)$.
\item {\bf Locality:}
For all sets $F \subseteq G \subseteq H$ with $f(F)=f(G)$ and every element $h \in H$, if $f(G) < f(G \cup \{h\})$ then $f(F) < f(F \cup \{h\})$.
\end{itemize}
A minimal subset $B \subseteq H$ with $f(B') < f(B)$ for
all proper subsets $B'$ of $B$ is called a \emph{basis} of $H$. An
\emph{optimal basis} is a basis $B$ with $f(B)=f(H)$. The maximum cardinality
of a basis is called the \emph{(combinatorial) dimension} of $(H,f)$ and
denoted by $\dim(H,f)$. LP-type problems cover many important optimization problems.

\subsubsection*{Linear optimization}

In this case, $H$ is the set of all linear constraints and $f(H)$ denotes the
optimal value in the polytope formed by $H$ with respect to the given objective
function. W.l.o.g., we may assume that $(H,f)$ is non-degenerate, i.e., for every subset $G \subseteq H$, $f(G)$ is associated with a unique solution (by, for example, slightly perturbing the coefficients in the linear constraints).
The monotonicity condition obviously holds in this case. Also, the
locality condition holds since if $f(G) < f(G \cup \{h\})$ (i.e., if $h$ is
violated by the solution associated with $f(G)$), then due to $f(F)=f(G)$, $h$ is also violated by $f(F)$. The combinatorial dimension is simply the number of variables of the LP. 

\subsubsection*{Smallest enclosing ball}

In this case, $H$ is a set of points in a Euclidean space and $f(H)$ denotes the radius of the smallest enclosing ball for $H$. The monotonicity condition can be verified
easily. Also the locality condition holds since if the smallest enclosing
balls for $G$ and $F\subseteq G$ have the same radius (and thus they actually
are the same ball) and point $h$ lies outside of the ball of $G$, then $h$ must also lie outside of the ball of $F$. Since in the 2-dimensional case at most 3 points are sufficient to determine the smallest enclosing ball for $H$, the combinatorial
dimension of this problem is 3. For $d$ dimensions, at most $d+1$ points are
sufficient.

\bigskip

Clarkson \cite{Cla95} proposed a very elegant randomized algorithm for solving LP-type problems (see Algorithm~\ref{alg:Clarkson}). In this algorithm, each $h\in H$ has a multiplicity of $\mu_h \in \N$, and $H(\mu)$ is a multiset where each $h \in H$ occurs $\mu_h$ times in $H(\mu)$. The algorithm requires a subroutine for computing $f(S)$ for sets $S$ of size $O(\dim(H,f)^2)$, but this is usually straightforward if $\dim(H,f)$ is a constant. The runtime analysis is simple enough so that we will review it in this paper, since it will also be helpful for the analysis of our distributed algorithms. In the following, let $d = \dim(H,f)$, and we say that an iteration of the repeat-loop is {\em successful} if $|V| \le |H(\mu)|/(3d)$.

\begin{algorithm}[th]
  \begin{algorithmic}[1]
    \If{$|H| \le 6 \dim(H,f)^2$}
      \Return $f(H)$
    \Else
      \State $r:=6\dim(H,f)^2$
      \ForAll{$h \in H$} $\mu_h:=1$
      \EndFor
      \Repeat
        \State Choose a random multiset $R$ of size $r$ from $H(\mu)$
        \State $V:=\{ h \in H(\mu) \mid f(R)<f(R \cup \{h\}) \}$
        \If{$|V| \le |H(\mu)|/(3 \dim(H,f))$}
          \ForAll{$h \in V$} $\mu_h:=2\mu_h$
          \EndFor
        \EndIf
      \Until{$V= \emptyset$}
    \EndIf
    \State \Return $f(R)$
  \end{algorithmic}
  \caption{Clarkson Algorithm.}
  \label{alg:Clarkson}
\end{algorithm}

\begin{lemma}[\cite{GW01}] \label{lem_V}
Let $(H,f)$ be an LP-type problem of dimension $d$ and let $\mu$ be any multiplicity function. For any $1 \le r < m$, where $m=|H(\mu)|$, the expected size of $V_R=\{h \in H(\mu) \mid f(R)<f(R \cup \{h\}) \}$ for a random multiset $R$ of size $r$ from $H(\mu)$ is at most $d \cdot \frac{m-r}{r+1}$.
\end{lemma}
\begin{proof}
Let ${H(\mu) \choose r}$ be the set of all multisets of $r$ elements in $H(\mu)$,
i.e., all results for $R$. By definition of the expected value it holds
\[
  \E[|V_R|] = \frac{1}{{m \choose r}} \sum_{R \in {H(\mu) \choose r}} |V_R|
\]
For $R \in {H(\mu) \choose r}$ and $h \in H(\mu)$ let $X(R,h)$ be the indicator
variable for the event that $f(R)<f(R \cup \{h\})$. Then
we have
\begin{eqnarray*}
 {m \choose r} \E[|V_R|] & = & \sum_{R \in {H(\mu) \choose r}} |V_R|
   = \sum_{R \in {H(\mu) \choose r}} \sum_{h \in H(\mu) - R} X(R,h) \\
 & \stackrel{(1)}{=} & \sum_{Q \in {H(\mu) \choose r+1}} \sum_{h \in Q} X(Q-h ,h) \\
 & \stackrel{(2)}{\le} & \sum_{Q \in {H(\mu) \choose r+1}} d = {m \choose r+1} \cdot d
\end{eqnarray*}
Equation (1) is true since choosing a set $R$ of size $r$ from $H(\mu)$ and subsequently choosing some $h\in H(\mu)-R$ is the same as choosing a set $Q$ of $r+1$ constraints from $H(\mu)$ and the subsequent removal of $h$ from $Q$. Equation (2) follows from the fact that the dimension of $(H,f)$ --- and therefore also of $(Q,f)$ --- is at most $d$ and the monotonicity condition, which implies that there are at most $d$ many $h \in Q$ with $f(Q-h)<f(Q)$. 
Resolving the inequality to $\E[|V_R|]$ results in the lemma.
\end{proof}

From this lemma and the Markov inequality it immediately follows that the probability that $|V| > |H(\mu)|/(3d)$ is at most $1/2$. Moreover, it holds:

\begin{lemma}[\cite{Lit87,Wel88}] \label{lem_success}
Let $k \in \N$ and $B$ be an arbitrary optimal basis of $H$. After $k \cdot d$ successful iterations, $2^k \le \mu(B) < n \cdot e^{k/3}$.
\end{lemma}
\begin{proof}
Each successful iteration increases the multiplicity of $H$ by a factor of at
most $(1+1/(3d))$. Therefore, $\mu(B) \le \mu(H) \le n(1+1/(3d))^{k \cdot d} < n \cdot e^{k/3}$. On the other hand, for each successful iteration with
$V\neq \emptyset$, $f(R) < f(H) = f(B)$. Due to the monotonicity condition, $f(R)<f(R \cup B)$ and $f(R) \le f(R \cup B') \le f(R \cup B)=f(H)$ for any subset $B' \subseteq B$. Let $B'$ be any maximal subset of $B$ (w.r.t. $\subseteq$) with $f(R)=f(R \cup B')$. Since $B' \subset B$, there is an $h \in B \setminus B'$ with $f(R \cup B')<f(R \cup B' \cup \{h\})$ and therefore, due to the locality condition, $f(R)<f(R \cup \{h\})$. Hence, there is a constraint in $B$ that is doubled at least $k$ times in $k \cdot d$ successful iterations, which implies that $\mu(B) \ge 2^k$.
\end{proof}

Lemma~\ref{lem_success} implies that Clarkson's algorithm must terminate after at most $O(d \log n)$ successful iterations (as otherwise $2^k > n \cdot e^{k/3}$), so Clarkson's algorithm performs at most $O(d \log n)$ iterations of the repeat-loop, on expectation. This bound is also best possible in the worst case for any $d \ll n$: given that there is a unique optimal basis $B$ of size $d$, its elements can have a multiplicity of at most $\sqrt{n}$ after $(\log n)/2$ iterations, so the probability that $B$ is contained in $R$ is polynomially small in $n$ up to that point.

Clarkson's algorithm has the advantage that it can easily be transformed into a distributed algorithm with expected runtime $O(d \log^2 n)$ if $n$ nodes are available that are interconnected by a hypercube, for example, because in that case every round of the algorithm can be executed in $O(\log n)$ communication rounds w.h.p.\footnote{By ``with high probability'', or short, ``w.h.p.'', we mean a probability of least $1-1/n^c$ for any constant $c>0$.}. However, it has been completely open so far whether it is also possible to construct a distributed algorithm for LP-type problems with an expected runtime of $O(d \log n)$ (either with a variant of Clarkson's algorithm or a different approach). We will show in this paper that this is possible when running certain variants of Clarkson's algorithm in the gossip model, even if $H$ has a polynomial size.

\subsection{Network Model}

We assume that we are given a fixed node set of size $n$ consisting of the nodes $v_1,\ldots,v_n$. In our paper, we do not require the nodes to have IDs since all of our protocols work for fully anonymous nodes. Moreover, we assume the standard synchronous message passing model, i.e., the nodes operate in synchronous {\em (communication) rounds}, and all messages sent (or requested) in round $i$ will be received at the beginning of round $i+1$. 

In the (uniform) gossip model, a node can only send or receive messages via random {\em push} and {\em pull} operations. In a push operation, it can send a message to a node chosen uniformly at random while in a pull operation, it can ask a node chosen uniformly at random to send it a message. We will restrict the message size (i.e., its number of bits) to $O(\log n)$. A node may execute multiple push and pull operations in parallel in a round. The number of push and pull operations executed by it is called its {\em (communication) work}.

Protocols for the gossip model are usually very practical due to their fast convergence, their simplicity, and their stability under stress and disruptions. Many gossip-based protocols have already been presented in the past, including protocols for information dissemination, network coding, load-balancing, consensus, and quantile computations (see \cite{DGMSS11,Haeu16,HMS18,KSSV00,KDG03} for some examples). Also, gossip protocols can be used efficiently in the context of population protocols and overlay networks, two important areas of network algorithms. In fact, it is easy to see that any algorithm with runtime $T$ and maximum work $W$ in the gossip model can be emulated by overlay networks in $O(T+\log n)$ time and with maximum work $O(W \log n)$ w.h.p. (since it is easy to set up (near-)random overlay edges in hypercubic networks in $O(\log n)$ time).

\subsection{Related Work}

There has already been a significant amount of work on finding efficient sequential and parallel algorithms for linear programs of constant dimension (i.e., a constant number of variables), which is a special case of LP-type problems of constant combinatorial dimension (see \cite{DGMW17} for a very thorough survey). We just focus here on parallel and distributed algorithms.
The fastest parallel algorithm known for the CRCW PRAM is due to Alon and Megiddo \cite{AM94}, which has a runtime of $O(d^2 \log^2 d)$. It essentially follows the idea of Clarkson, with the main difference that it replicates elements in $V$ much more aggressively by exploiting the power of the CRCW PRAM. This is achieved by first compressing the violated elements into a small area and then replicating them by a factor of $n^{1/(4d)}$ (instead of just 2). The best work-optimal algorithm for the CRCW PRAM is due to Goodrich \cite{Goo93}, which is based on an algorithm by Dyer and Frieze \cite{DF89} and has a runtime of $O((\log \log n)^d)$. This also implies a work-optimal algorithm for the EREW PRAM, but the runtime increases to $O(\log n (\log \log n)^d)$ in this case. The fastest parallel algorithm known for the EREW PRAM is due to Dyer 
\cite{Dye95}, which achieves a runtime of $O(\log n (\log \log n)^{d-1})$ when using an $O(\log n)$-time parallel sorting algorithm (like Cole's algorithm). Since the runtime of any algorithm for solving a linear program of constant dimension in an EREW PRAM is known to be $\Omega(\log n)$ \cite{DGMW17}, the upper bound is optimal for $d=1$. 

Due to Ranade's seminal work \cite{Ran91}, it is known that any CRCW PRAM step can be emulated in a butterfly network in $O(\log n)$ communication rounds, yielding an $O(\log n)$-time algorithm for linear programs of constant dimension in the butterfly. However, 
it is not clear whether any of the parallel algorithms would work for arbitrary LP-type problems. Also, none of the proposed parallel algorithms seem to be easily adaptable to an algorithm that works efficiently (i.e., in time $o(\log^2 n)$ and with $polylog$ work) for the gossip model as they require processors to work together in certain groups or on certain memory locations in a coordinated manner, and assuming anonymous nodes would further complicate the matter.

Algorithms for (integer) linear programs have also been investigated in the distributed domain. Their study was initiated by Papadimitriou and Yannakakis \cite{PY93}. 
Bartal, Byers and Raz \cite{BBR04} presented a distributed approximation scheme for positive linear programs with a polylogarithmic runtime. Kuhn, Moscibroda, and Wattenhofer \cite{KMW06} present a distributed approximation scheme for packing LPs and covering LPs. For certain cases, their scheme is even a local approximation scheme, i.e., it only needs to know a constant-distance neighborhood and therefore can be implemented in a constant number of rounds (given sufficiently large edge bandwidths). Flor\'een et al. \cite{FKMS08} studied the problem of finding local approximation schemes for max-min linear programs, which are a generalized form of packing LPs. They show that in most cases there is no local approximation scheme and identify certain cases where a local approximation scheme can be constructed. Positive LPs and max-min LPs are a special form of LP-type problems, but to the best of our knowledge, no distributed algorithms have been formally studied for LP-type problems in general.

As mentioned above, LP-type problems were introduced by Sharir and Welzl \cite{SW92}. Since then, various results have been shown, but only for sequential algorithms. Combining results by G\"artner \cite{Gar95} with Clarkson's methods, G\"artner and Welzl \cite{GW96} showed that an expected linear number of violation tests and basis computations is sufficient to solve arbitrary LP-type problems of constant combinatorial dimension. For the case of a large combinatorial dimension $d$, Hansen and Zwick \cite{HZ15} proposed an algorithm with runtime $e^{O(\sqrt{d})}$. Extensions of LP-type problems were studied by G\"artner \cite{Gar95} (abstract optimization problems) 
and Skovron \cite{Sko07} (violator spaces). G\"artner et al. \cite{GMRS08} and Brise and G\"artner \cite{BG11} showed that Clarkson's approach still works for violator spaces.

For the applications usually considered in the context of LP-type problems, the combinatorial dimension is closely related to the minimum size of an optimal basis. But there are also LP-type problems whose combinatorial dimension might be much larger. Prominent examples are the hitting set problem and the equivalent set cover problem. Both are known to be NP-hard problems. Also, Dinur and Steurer \cite{DS14} showed that the set cover problem (and therefore the hitting set problem) cannot be approximated within a $(1-o(1))\ln n$ factor unless $P=NP$. Based on Clarkson's algorithm, Br\"onnimann and Goodrich \cite{BG95} and Agarwal and Pan \cite{AP14} gave algorithms that compute an approximate set cover / hitting set of a geometric set cover instance in $O(n \cdot polylog(n))$ time. Their results imply, for example, an $O(\log \log OPT)$-approximation algorithm for the hitting set problem that runs in $O(n \log^3 n \log \log \log OPT)$ time for range spaces induced by 2D axis-parallel rectangles, and an $O(1)$-approximate set cover in $O(n \log^4 n)$ time for range spaces induced by 2D disks. The currently best distributed algorithm for the set cover problem was presented by Even, Ghaffari, and Medina \cite{EGM18}. They present a deterministic distributed algorithm for computing an $f(1+\epsilon)$-approximation of the set cover problem, for any constant $\epsilon>0$, in $O(\log (f \Delta)/\log \log (f \Delta))$ rounds, where $f$ is the maximum element frequency and $\Delta$ is the cardinality of the largest set. This almost matches the $\Omega(\log(\Delta)/\log \log (\Delta))$ lower bound of Kuhn, Moscibroda, and Wattenhofer \cite{KMW16}. Their network is a bipartite graph with the nodes on one side representing the elements and the nodes on the other side representing the sets, and their algorithm follows the primal-dual scheme. However, it is unclear whether their algorithm can be adapted to the gossip model.


\subsection{Our Results} \label{sec:new}

In all of our results, we assume that initially, $H$ is randomly distributed among the nodes. This is easy to achieve in the gossip model if this is not the case (for example, each node initially represents its own point for the smallest enclosing ball problem) by performing a push operation on each element. The nodes are assumed to know $f$, and we require the nodes to have a constant factor estimate of $\log n$ for the algorithms to provide a correct output, w.h.p., but they may not have any information about $|H|$. For simplicity, we also assume that the nodes know $d$. If not, they may perform a binary search on $d$ (by stopping the algorithm if it takes too long for some $d$ to switch to $2d$), which does not affect our bounds below since they depend at least linearly on $d$.

We usually assume that the dimension $d$ of the given LP-type problem is a constant (i.e., independent of $n$), though our proofs and results would also be true for non-constant $d$ (as long as $d$ is sufficiently small compared to $n$). In Section 2, we start with the lightly loaded case (i.e., $|H|=O(n \log n)$) and prove the following theorem.

\begin{theorem} \label{th:main1}
For any LP-type problem $(H,f)$ satisfying $|H|=$ $O(n \log n)$, the Low-Load Clarkson Algorithm finds an optimal solution in $O(d \log n)$ rounds with maximum work $O(d^2 + \log n)$ per round, w.h.p.
\end{theorem}

At a high level, the Low-Load Clarkson Algorithm is similar to the original Clarkson algorithm, but sampling a random multiset and termination detection are more complex now, and a filtering approach is needed to keep $|H(\mu)|$ low at all times so that the work is low. In Section 3, we then consider the highly loaded case and prove the following theorem.

\begin{theorem} \label{th:main2}
For any LP-type problem $(H,f)$ with $|H|=\omega(n \log n)$ and $|H|=poly(n)$, the High-Load Clarkson Algorithm finds an optimal solution in $O(d \log n)$ rounds with maximum work $O(d\log n)$ per round, w.h.p. If we allow a maximum work of $O(d \log^{1+\epsilon} n)$ per round, for any constant $\epsilon>0$, the runtime reduces to $O(d \log(n)/\log\log(n))$, w.h.p.
\end{theorem}

Note that as long as we only allow the nodes to spend polylogarithmic work per round, a trivial lower bound on the runtime when using Clarkson's approach is $\Omega(\log(n)/\log\log(n))$ since in $o(\log(n)/\log\log(n))$ rounds an element in $H$ can only be spread to $n^{o(1)}$ nodes, so the probability of fetching it under the gossip model is minute.

The reason why we designed different algorithms for the lightly loaded and highly loaded cases is that the Low-Load Clarkson Algorithm is much more efficient than the High-Load Clarkson Algorithm concerning internal computations. Also, it is better concerning the work for the lightly loaded case, but its work does not scale well with an increasing $|H|$. The main innovation for Theorem~\ref{th:main2} is that we come up with a Chernoff-style bound for $|V|$ that holds for all LP-type problems. G\"artner and Welzl \cite{GW01} also provided a Chernoff-style bound on $|V|$ for LP-type problems, but their proof only works for LP-type problems that are regular (i.e., for all $R \subseteq H$ with $|R| \ge d$, all optimal bases of $R$ have a size of exactly $d$) and non-degenerate (i.e., every $R \subset H$ with $|R| \ge d$ has a unique optimal basis). While regularity can be enforced in the non-degenerate case, it is not known so far how to make a general LP-type problem non-degenerate without substantially changing its structure (though for most of the applications considered so far for LP-type problems, slight perturbations of the input would solve this problem). Since the duplication approach of Clarkson's algorithm generates degenerate instances, their Chernoff-style bound therefore cannot be used here.

Finally, we will study two LP-type problems that can potentially have a very high combinatorial dimension even though the size of an optimal basis might just be a constant: the hitting set problem and the set cover problem.

Let $X=\{1,\ldots,n\}$ be a set of elements and ${\mathcal S}$ be a collection of subsets of $X$. A subset $H \subseteq X$ is called a {\em hitting set} of ${\mathcal S}$ if for all $S \in {\mathcal S}$, $S \cap H \not= \emptyset$. 
In the {\em hitting set problem} we are given $(X,{\mathcal S})$, and the goal is to find a hitting set of minimum size. 

First of all, it is easy to verify that $(X,f)$, where $f(U)$ for any subset $U$ of $X$ denotes the number of sets in ${\mathcal S}$ intersected by $U$, satisfies the monotonicity and locality conditions, so $(X,f)$ is an LP-type problem. However, its combinatorial dimension might be much larger than the size of a minimum hitting set. Nevertheless, we present a distributed gossip-based algorithm that is able to find an approximate solution efficiently. 

We assume that every node knows ${\mathcal S}$ so that it can locally evaluate $f$. Note that knowing ${\mathcal S}$ may not necessarily mean that every node knows $X$ because the sets might just be defined implicitly w.r.t. $X$, e.g., the sets $S \in {\mathcal S}$ might represent polygons in some 2-dimensional space.  Also, initially, the points in $X$ are randomly distributed among the nodes.
Under these assumptions, we can show the following theorem.

\begin{theorem} \label{th:main3}
For any hitting set problem $(X,{\mathcal S})$ with $|X|=n$ and $|{\mathcal S}|=s$ and a minimum hitting set of size $d$, our Hitting Set Algorithm finds a hitting set of size $O(d \log (ds))$ in $O(d \log n)$ rounds with maximum work $O(d \log(ds) + \log n)$ per round, w.h.p.
\end{theorem}

Finally, let us review the set cover problem. Again, let $X=\{1,\ldots,n\}$ be a set of elements and ${\mathcal S}$ be a collection of subsets of $X$, where we assume here that $\bigcup_{S \in {\mathcal S}} S = X$. A set ${\mathcal C} \subseteq {\mathcal S}$ is called a {\em set cover} of $X$ if $\bigcup_{S \in {\mathcal C}} S = X$. In the (simple form of the) {\em set cover problem} we are given $(X,{\mathcal S})$, and the goal is to find a set cover of minimum size, i.e., a minimum number of sets.

It is easy to verify that $({\mathcal S},f)$, where $f(U)$ for any subset $U$ of ${\mathcal S}$ denotes the number of elements in $X$ covered by $U$, satisfies the monotonicity and locality conditions, so $({\mathcal S},f)$ is an LP-type problem.

We assume that every node knows $X$ so that it can locally evaluate $f$, and initially the elements in ${\mathcal S}$ are randomly distributed among the nodes. Note that even though some set might contain many elements in $X$,
we will assume here that every $S \in {\mathcal S}$ has a compact representation (like a polygon) so that it can be sent in one message. 

We can then use our Hitting Set Algorithm to solve any set cover problem with the same bounds as in Theorem~\ref{th:main3}, because there is a well-known equivalent formulation as a hitting set problem: Given that ${\mathcal S}=\{S_1,\ldots,S_s\}$, let $Y=\{1,\ldots,s\}$ and ${\mathcal M} = \{M_1,\ldots,M_n\}$, where $M_i = \{ j \in \{1,\ldots,s\} \mid i \in S_j\}$. Then a set cover in $(X,{\mathcal S})$ corresponds to a hitting set in $(Y,{\mathcal M})$.

\section{Low-Load Clarkson Algorithm}

Suppose that we have an arbitrary LP-type problem $(H,f)$ of dimension $d$ with $|H|=O(n \log n)$. First, we present and analyze an algorithm for $|H| \ge n$, and then we extend it to any $1 \le |H|=O(n \log n)$.

Recall that initially the elements of $H$ are assigned to the nodes uniformly and independently at random. Let us denote the set of these elements in node $v_i$ by $H_0(v_i)$ to distinguish them from copies created later by the algorithm, and let $H_0 = \bigcup_i H_0(v_i)$.

At any time, $H(v_i)$ denotes the (multi)set of elements in $H$ known to $v_i$ (including the elements in $H_0$) and $H(V)=\bigcup_i H(v_i)$, where $V$ represents the node set. Let $m=|H(V)|$. At a high level, our distributed algorithm is similar to the original Clarkson algorithm, but sampling a random multiset and termination detection are more complex now (which will be explained in dedicated subsections). In fact, the sampling might fail since a node $v_i$ might not be able to collect enough elements for $R_i$. Also, a filtering approach is needed to keep $|H(V)|$ low at all times (see Algorithm~\ref{alg:DistClarkson}). However, it will never become too low since the algorithm never deletes an element in $H_0$, so $|H(V)| \ge n$ at any time. Note that never deleting an element in $H_0$ also guarantees that no element in $H$ will ever be washed out (which would result in incorrect solutions). 

For the runtime analysis, we note that sampling $R_i$ can be done in one round (see Section~\ref{sec:sampling}), spreading $W_i$ just takes one round (by executing the push operations in parallel), and we just need one more round for processing the received elements $h$, so for simplicity we just assume in the following that an iteration of the repeat loop takes one round. We start with a slight variant of Lemma~\ref{lem_V}.

\begin{algorithm}[th]
  \begin{algorithmic}[1]
    \Repeat
      \ForAll{nodes $v_i$ in parallel}
        \State choose a random multiset $R_i$ of size $6d^2$ from $H(V)$
        \If{the sampling of $R_i$ succeeds}
          \State $W_i:=\{ h \in H(v_i) \mid f(R_i)<f(R_i \cup \{h\}) \}$
          \ForAll{$h \in W_i$} {\bf push}($h$) {\em // randomly spread $W_i$}
          \EndFor
        \EndIf
        \ForAll{$h$ received by $v_i$} add $h$ to $H(v_i)$
        \EndFor
        \ForAll{$h \in H(v_i) - H_0(v_i)$} 
          \State keep $h$ with probability $1/(1+1/(2d))$
        \EndFor
      \EndFor
    \Until{at least one $v_i$ satisfies $f(R_i)=f(H)$}
  \end{algorithmic}
  \caption{Low-Load Clarkson Algorithm.}
  \label{alg:DistClarkson}
\end{algorithm}

\begin{lemma} \label{lem_Wi}
Let $(H,f)$ be an LP-type problem of dimension $d$ and let $H(V)$ be any multiset of $H$ of size $m$. For any $1 \le r < m$, the expected size of $W_i=\{h \in H(v_i) \mid f(R)<f(R \cup \{h\}) \}$ for a random multiset $R$ of size $r$ from $H(V)$ is at most $d \cdot \frac{m-r}{n(r+1)}$.
\end{lemma}
\begin{proof}
According to Lemma~\ref{lem_V}, the expected size of $W(R)=\{ h \in H(V) \mid f(R)<f(R \cup \{h\}) \}$ for a random multiset $R$ is at most $d \cdot \frac{m-r}{r+1}$. Since every element in $H(V)$ has a probability of $1/n$ to belong to $H(v_i)$, $\E[|W_i|] \le d \cdot \frac{m-r}{n(r+1)}$.
\end{proof}

This allows us to prove the following lemma.

\begin{lemma} \label{lem:Wgrowth}
For all $i$, $|W_i|=O(m/n + \log n)$, w.h.p., and $\sum_{i=1}^n |W_i|$ $\le m/(3d)$, w.h.p. 
\end{lemma}
\begin{proof}
Let the random variable $X_i$ be defined as $|W_i|$ and let $X = \sum_i X_i$. If the sampling of $R_i$ fails then, certainly, $X_i=0$, and otherwise, $\E[X_i] \le d \cdot \frac{m-r}{n(r+1)}$ for all $i$. Thus, $\E[X] \le d \cdot \frac{m-r}{r+1}$. Also, since the elements in $H(V)$ are distributed uniformly and independently at random among the nodes at all times, the standard Chernoff bounds imply that $|H(v_i)|=O(m/n + \log n)$ w.h.p., and therefore also $X_i \le O(m/n + \log n)$ w.h.p. 
Unfortunately, the $X_i$'s are not independent since $H(v_i)$ is not chosen independently of the other $H(v_j)$'s, but the dependencies are minute: given that we have already determined $H(v_j)$ for $k$ many $v_j$'s, where $k=o(n)$ is sufficiently small, the probability that any one of the remaining elements $h\in H(V)$ is assigned to $H(v_i)$ is $1/(n-o(n)) = (1+o(1))/n$, so that for any subset $S \subseteq \{1,\ldots,n\}$ of size $k$,
\[
  \E[\prod_{i \in S} X_i] \le \left( (1+o(1))d \cdot \frac{m-r}{n(r+1)} \right)^k \qquad (*)
\]
This allows us to use a Chernoff-Hoeffding-style bound for $k$-wise negatively correlated random variables, which is a slight extension of Theorem 3 in \cite{SSS95}:

\begin{theorem} \label{th:chernoff}
Let $X_1,\ldots,X_n$ be random variables with $X_i \in [0,C]$ for some $C>0$. Suppose there is a $k>1$ and $q>0$ with $\E[\prod_{i \in S} X_i] \le q^s$ for all subsets $S \subseteq \{1,\ldots,n\}$ of size $s \le k$. Let $X = \sum_{i=1}^n X_i$ and $\mu = q \cdot n$. Then it holds for all $\delta > 0$ with $k \ge \lceil \mu \delta \rceil$ that
\[
  \Pr[X \ge (1+\delta)\mu] \le e^{-\min\{\delta^2, \delta\} \mu/(3C)}
\]
\end{theorem}

Setting $C=\Theta(m/n + \log n)$, $\mu=q \cdot n$ with $q=(1+o(1))d \cdot \frac{m-r}{n(r+1)}$ and $r=6d^2$, and $\delta>0$ large enough so that $\delta^2 \mu = \omega(C \ln n)$ but $\delta \mu = o(n)$ so that inequality $(*)$ applies, which works for $\delta = \Theta(\sqrt{(C d \ln n)/m})=O(\sqrt{(d \ln^2 n)/n})$, $\Pr[X \ge m/(3d)]$ is polynomially small in $n$. 
\end{proof}

Next, we show that $m$ will never be too large, so that the communication work of the nodes will never be too large.

\begin{lemma} \label{lem:Hsize}
For up to polynomially many rounds of the Low-Load Clarkson Algorithm, $|H(V)|=O(|H_0|)$, w.h.p.
\end{lemma}
\begin{proof}
Let $q=|H(V)-H_0|$ and suppose that $|H(V)| \ge c|H_0|$ for some $c\ge 4$, which implies that $q \ge (c-1)/c \cdot m$. Then it holds for the size $q'$ of $H(V)-H_0$ at the end of a repeat-round that
\begin{eqnarray*}
  \E[q'] & \le & (q+m/(3d)) \cdot 1/(1+1/(2d)) \\
  & \le & (q+(c q)/((c-1) \cdot 3d)) \cdot 1/(1+1/(2d)) \\
  & = & q (1 - (1/(2d)-c/(3(c-1)d))/(1+1/(2d))) \\
  & = & q (1-\Theta(1/d))
\end{eqnarray*}
Since the decision to keep elements $h \in H(V)-H_0$ is done independently for each $h$, it follows from the Chernoff bounds that $\Pr[q'>q]$ is polynomially small in $n$ for $|H(V)| \ge 4|H_0|$. Moreover, Lemma~\ref{lem:Wgrowth} implies that $|H(V)|$ can increase by at most $|H(V)|/3$ in each round, w.h.p., so for polynomially many rounds of the algorithm, $|H(V)| \le 5|H_0|$ w.h.p.
\end{proof}

Thus, combining Lemma~\ref{lem:Wgrowth} and Lemma~\ref{lem:Hsize}, the maximum work per round for pushing out some $W_i$ is bounded by $O(\log n)$ w.h.p. Next we prove a lemma that adapts Lemma~\ref{lem_success} to our setting.

\begin{lemma} \label{lem:llsuccess}
Let $B$ be an arbitrary optimal basis of $H$. If, for $T$ many rounds of the Low-Load Clarkson Algorithm, every node was successful in sampling a random multiset and no $v_i$ satisfies $f(R_i)=f(H)$, then $\E[|\{ h \in H(V) \mid h \in B\}|] \ge (2/\sqrt{e})^{T/d}$ after these $T$ rounds.
\end{lemma}
\begin{proof}
Let $B=\{h_1(B),\ldots,h_b(B)\}$, $b \le d$, and let $p_{i,j}$ be the probability that $f(R_i)<f(R_i \cup \{h_j(B)\})$. If node $v_i$ has chosen some $R_i$ with $f(R_i)<f(H)$, then there must exist an $h_j(B)$ with $f(R_i)<f(R_i \cup \{h_j(B)\})$, which implies that under the condition that $f(R_i)<f(H)$, $\sum_j p_{i,j} \ge 1$. The $p_{i,j}$'s are the same for each $v_i$ since each $v_i$ has the same probability of picking some multiset $R$ of $H(V)$ of size $6d^2$. Hence, we can simplify $p_{i,j}$ to $p_j$ and state that $\sum_j p_j \ge 1$. Now, let $p_{j,t}$ be the probability that $f(R)<f(R \cup \{h_j(B)\})$ for a randomly chosen multiset $R$ in round $t$, and fix any values for the $p_{j,t}$ so that $\sum_j p_{j,t} \ge 1$ for all $j$ and $t$. Let $\mu_{j,t}$ be the multiplicity of $h_j(B)$ at the end of round $t$. Then, for all $j$, $\mu_{j,0} \ge 1$, and
\[
  \E[\mu_{j,t+1}] \ge \frac{1+p_{j,t}}{1+1/(2d)} \cdot \mu_{j,t}
\]
Hence, $\E[\mu_{j,T}] \ge (\prod_{t=1}^T (1+p_{j,t}))/(1+1/(2d))^T$. Since $1+x \ge 2^x$ for all $x \in [0,1]$, it follows that $\prod_t (1+p_{j,t}) \ge \prod_t 2^{p_{j,t}} = 2^{\sum_t p_{j,t}}$. Also, since $\sum_{t=1}^T \sum_j p_{j,t} \ge T$, there must be a $j^*$ with $\sum_{t=1}^T p_{j^*,t} \ge T/d$.
Therefore, there must be a $j^*$ with $\E[\mu_{j^*,T}] \ge 2^{T/d}/(1+1/(2d))^T \ge 2^{T/d}/e^{T/(2d)}$, which completes the proof.
\end{proof}

Since $|H(V)|$ is bounded by $O(|H_0|)$ w.h.p., the expected number of copies of $h_{j^*}(B)$ should be at most $O(|H_0|)$ w.h.p. as well. Due to Lemma~\ref{lem:Hsize}, this cannot be the case if $T=\Omega(d \log n)$ is sufficiently large. Thus, the algorithm must terminate within $O(d \log n)$ rounds w.h.p.

In order to complete the description of our algorithm, we need distributed algorithms satisfying the following claims:
\begin{enumerate}
\item The nodes $v_i$ succeed in sampling multisets $R_i$ uniformly at random in a round, w.h.p., with maximum work $O(d^2+\log n)$. 
\item Once a node $v_i$ has chosen an $R_i$ with $f(R_i)=f(H)$, all nodes are aware of that within $O(\log n)$ communication rounds, w.h.p., so that the Low-Load Clarkson Algorithm can terminate. The maximum work for the termination detection is $O(\log n)$ per round.
\end{enumerate}
The next two subsections are dedicated to these algorithms.

\subsection{Sampling random multisets} \label{sec:sampling}

For simplicity, we assume here that every node knows the exact value of $\log n$, but it is easy to see that the sampling algorithm also works if the nodes just know a constant factor estimate of $\log n$, if the constant $c$ used below is sufficiently large.

Each node $v_i$ samples a multiset $R_i$ in a way that is as simple as it can possibly get: $v_i$ asks $s=c(6d^2 + \log n)$ random nodes $v_j$ via pull operations to send it a random element in $H(v_j)$, where $c$ is a sufficiently large constant. Out of the returned elements, $v_i$ selects $6d^2$ distinct elements at random for its multiset $R_i$. If $v_i$ hasn't received at least $6d^2$ distinct elements, the sampling fails. Certainly, the work for each node is just $O(d^2 + \log n)$.

\begin{lemma}
For any $i$, node $v_i$ succeeds in sampling a multiset $R_i$ uniformly at random, w.h.p.
\end{lemma}
\begin{proof}
Suppose that $v_i$ succeeds in receiving $k$ distinct elements in the sampling procedure above. Since the elements in $H(V)$ are distributed uniformly and independently at random among the nodes, every multiset $R$ of size $k$ in $H(v)$ has the same probability of representing these $k$ elements. Hence, it remains to show that $v_i$ succeeds in receiving at least $6d^2$ elements w.h.p.

Consider any numbering of the pull requests from 1 to $s$. For the $j$th pull request of $v_i$, two bad events can occur. First of all, the pull request might be sent to a node that does not have any elements. Since $|H(V)| \ge n$, the probability for that is at most $(1-1/n)^n \le 1/e$. Second, the pull request might return an element that was already returned by one of the prior $j-1$ pull requests. Since this is definitely avoided if the $j$th pull request selects a node that is different from the nodes selected by the prior $j-1$ pull requests, the probability for that is at most $(j-1)/n$. So altogether, the probability that a pull request fails is at most $1/e + s/n \le 1/2$. 

Now, let the binary random variable $X_j$ be 1 if and only if the $j$th pull request fails. Since the upper bound of $1/2$ for the failure holds independently of the other pull requests, it holds for any subset $S \subseteq \{1,\ldots,s\}$ that $\E[ \prod_{j \in S} X_j] \le (1/2)^{|S|}$. Hence, Theorem~\ref{th:chernoff} implies that $\sum_j X_j \le 3s/4$ w.h.p. If $c$ is sufficiently large, then $3s/4 \ge 6d^2$, which completes the proof.
\end{proof}

Note that our sampling strategy does not reveal any information about which elements are stored in $H(v_i)$, so each element still has a probability of $1/n$ to be stored in $H(v_i)$, which implies that Lemma~\ref{lem_Wi} still holds.

\subsection{Termination} \label{sec:termination}

We use the following strategy for each node $v_i$:

Suppose that in iteration $t$ of the repeat loop, $|W_i|=0$, i.e., $f(R_i)=f(R_i \cup H(v_i))$. Then $v_i$ determines an optimal basis $B$ of $R_i$, stores the entry $(t,B,1)$ in a dedicated set $S_i$, and performs a push operation on $(t,B,1)$. At the beginning of iteration $t_i$ of the repeat loop, $v_i$ works as described in Algorithm~\ref{alg:Termination}. In the comparison between $f(B')$ and $f(B)$ we assume w.l.o.g. that $f(B')=f(B)$ if and only if $B'=B$ (otherwise, we use a lexicographic ordering of the elements as a tie breaker). The parameter $c$ in the algorithm is assumed to be a sufficiently large constant known to all nodes.

\begin{algorithm}[th]
  \begin{algorithmic}[1]
    \ForAll{$(t,B,x)$ received by $v_i$}
      \If{there is some $(t,B',x')$ in $S_i$}
        \If{$f(B')>f(B)$} discard $(t,B,x)$
        \EndIf
        \If{$f(B')<f(B)$} replace $(t,B',x')$ by $(t,B,x)$
        \EndIf
        \If{$f(B')=f(B)$} 
          \State replace $(t,B',x')$ by $(t,B,\min\{x,x'\})$
        \EndIf
      \Else
        \State add $(t,B,x)$ to $S_i$
      \EndIf
    \EndFor
    \ForAll{$(t,B,x)$ in $S_i$}
      \If{$f(B)<f(B \cup H(v_i))$} $x:=0$ {\em // $B$ is invalid}
      \EndIf
      \If{$t < t_i-c \log n$} {\em // $B$ is mature}
        \State remove $(t,B,x)$ from $S_i$
        \If{$x=1$} output $f(B)$, stop
        \EndIf
      \Else
        \State {\bf push}($t,B,x$) {\em // a copy of $(t,B,x)$ is pushed out}
      \EndIf
    \EndFor
  \end{algorithmic}
  \caption{One round of the Termination Algorithm for $v_i$.}
  \label{alg:Termination}
\end{algorithm}

\begin{lemma}
If the constant $c$ in the termination algorithm is large enough, it holds w.h.p.:
Once a node $v_i$ satisfies $f(R_i)=f(H)$, then all nodes $v_j$ output a value $f(B)$ with $f(B)=f(H)$ after $c \log n$ iterations, and if a node $v_i$ outputs a value $f(B)$, then $f(B)=f(H)$.
\end{lemma}
\begin{proof}
Using standard arguments, it can be shown that if the constant $c$ is large enough, then for every iteration $t$, it takes at most $(c/2) \log n$ iterations, w.h.p., until the basis $B$ with maximum $f(B)$ injected into some $S_i$ at iteration $t$ (which we assume to be unique by using some tie breaking mechanism) is contained in all $S_i$'s. At this point, we have two cases. If $f(B)=f(H)$, then for all $v_i$, $f(B)=f(B \cup H(v_i))$ at any point from iteration $t$ to $t+(c/2)\log n$, and otherwise, there must be at least one $v_i$ at iteration $t+(c/2)\log n$ with $f(B)<f(B \cup H(v_i))$. In the first case, no $v_i$ will ever set $x$ in the entry $(t,B,x)$ to 0, so after an additional $(c/2)\log n$ iterations, every $v_i$ still stores $(t,B,1)$ and therefore outputs $B$. In the second case, there is at least one entry of the form $(t,B,0)$ at iteration $s+(c/2)\log n$. For this entry, it takes at most $(c/2) \log n$ further iterations, w.h.p., to spread to all nodes so that at the end, no node outputs $B$.
\end{proof}

Since the age of an entry is at most $c \log n$ and for each age a node performs at most one push operation, every node has to execute just $O(\log n)$ push operations in each round.

\subsection{Extension to any $|H|\ge 1$}

If $|H|<n$, the probability that our sampling strategy might fail will get too large. Hence, we need to extend the Low-Load Clarkson algorithm so that we quickly reach a point where $|H(V)| \ge n$ at any time afterwards. We do this by integrating a so-called {\em pull phase} into the algorithm.

Initially, a node $v_i$ sets its Boolean variable $pull$ to $true$ if and only if $H_0(v)=\emptyset$ (which would happen if none of the elements in $H$ has been assigned to it). Afterwards, it executes the algorithm shown in Algorithm~\ref{alg:DistClarkson1}. As long as $pull=true$ (i.e., $v_i$ is still in its pull phase), $v_i$ keeps executing a pull operation in each iteration of the algorithm, which asks the contacted node $v_j$ to send it a copy of a random element in $H_0(v_j)$, until it successfully receives an element $h$ that way.
Once this is the case, $v_i$ pushes the successfully pulled element to a random node $v_j$ (so that all elements are distributed uniformly and independently at random among the nodes), which will store it in $H_0(v_j)$, and starts executing the Low-Load Clarkson algorithm from above.

\begin{algorithm}[th]
  \begin{algorithmic}[1]
    \Repeat
      \If{$pull=true$}  // {\em $v_i$ is still in its pull phase}
        \State {\bf pull}($h$) // {\em $v_i$ expects some $h \in H_0$}
        \If{$h\not= NULL$}
          \State {\bf push}($h,0$)
          \State $pull:=false$
        \EndIf
      \Else
        \State choose a random multiset $R_i$ of size $6d^2$ from $H(V)$
        \If{if the selection of $R_i$ succeeds}
          \State $W_i:=\{ h \in H(v_i) \mid f(R_i)<f(R_i \cup \{h\}) \}$
          \ForAll{$h \in W_i$} {\bf push}($h$)
          \EndFor
        \EndIf
      \EndIf
      \ForAll{$(h,0)$ received by $v_i$} add $h$ to $H_0(v_i)$
      \EndFor
      \ForAll{$h$ received by $v_i$} add $h$ to $H(v_i)$
      \EndFor
      \ForAll{$h \in H(v_i) - H_0(v_i)$} 
        \State keep $h$ with probability $1/(1+1/(2d))$
      \EndFor
    \Until{$v_i$ outputs a solution}
  \end{algorithmic}
  \caption{Extended Low-Load Clarkson Algorithm for $v_i$.}
  \label{alg:DistClarkson1}
\end{algorithm}

\begin{lemma}
After $O(\log n)$ rounds, all nodes have completed their pull phase, w.h.p. 
\end{lemma}
\begin{proof}
Note that no node will ever delete an element in $H_0$, and pull requests only generate elements for $H_0$, so the filtering approach of the Low-Load Clarkson algorithm cannot interfere with the pull phase. Thus, it follows from a slight adaptation of proofs in previous work on gossip algorithms (e.g., \cite{KSSV00}) that for any $|H| \ge 1$, all nodes have completed their pull phase after at most $O(\log n)$ rounds, w.h.p. 
\end{proof}

Certainly, $|H_0| \le n+|H| = O(n \log n)$ and $H \subseteq H_0$ at any time, and once all nodes have finished their pull phase, $|H_0| \ge n$, so we are back to the situation of the original Low-Load Clarkson Algorithm.

During the time when some nodes are still in their pull phase, some nodes might already be executing Algorithm~\ref{alg:DistClarkson}, which may cause the sampling of $R_i$ to fail for some nodes $v_i$. However, the analyses of Lemma~\ref{lem:Wgrowth} and Lemma~\ref{lem:Hsize} already take that into account. Once all nodes have finished their pull phase, Lemma~\ref{lem:llsuccess} applies, which means that after an additional $O(d \log n)$ rounds at least one node has found the optimal solution, w.h.p. Thus, after an additional $O(\log n)$ nodes, all nodes will know the optimal solution and terminate. Altogether, we therefore still get the same runtime and work bounds as before, completing the proof of Theorem~\ref{th:main1}.



\section{High-Load Clarkson Algorithm}

If $|H|= \omega(n \log n)$, then our LP-type algorithm in the previous section will become too expensive since, on expectation, $|W_i|$ might be in the order of $m/(dn)$, which is now $\omega(\log n)$. In this section, we present an alternative distributed LP-type algorithm that just causes $O(d \log n)$ work for any $|H|=poly(n)$, but the internal computations are more expensive then in the algorithm presented in the previous section. Again, we assume that initially the elements in $H$ are randomly distributed among the nodes in $V$. Let the initial $H(v_i)$ be all elements of $H$ assigned that way to $v_i$. As before, $H(V)= \bigcup_i H(v_i)$.

\begin{algorithm}[th]
  \begin{algorithmic}[1]
    \Repeat
      \ForAll{nodes $v_i$ in parallel}
        \State compute an optimal basis $B_i$ of $H(v_i)$
        \State {\bf push}($B_i$)
        \ForAll{$B_j$ received by $v_i$} 
          \State $W_j:=\{ h \in H(v_i) \mid f(B_j) < f(B_j \cup \{h\}) \}$
          \ForAll{$h \in W_j$} {\bf push}($h$)
          \EndFor
        \EndFor
        \ForAll{$h$ received by $v_i$} add $h$ to $H(v_i)$
        \EndFor
      \EndFor
    \Until{at least one $v_i$ satisfies $f(H(v_i))=f(H)$}
  \end{algorithmic}
  \caption{High-Load Clarkson Algorithm.}
  \label{alg:DistClarkson2}
\end{algorithm}

Irrespective of which elements get selected for the $W_i$'s in each round, $H(v_i)$ is a random subset of $H(V)$ because the elements in $H$ are assumed to be randomly distributed among the nodes and every element in $W_i$ is sent to a random node in $V$. Hence, if follows from $|H(V)|=\omega(n \log n)$ and the standard Chernoff bounds that $|H(v_i)|$ is within $(1 \pm \epsilon)|H(V)|/n$, w.h.p., for any constant $\epsilon>0$. Thus, we are computing bases of random multisets $R$ of size $r$ within $(1 \pm \epsilon)|H(V)|/n$, w.h.p.
This, in turn, implies with $\E[|W_i|] \le d \cdot \frac{m-r}{n(r+1)}$, where $m=|H(V)|$, that $\E[|W_i|] \le (1+\epsilon) d$. In the worst case, however, $|W_i|$ could be very large, so just bounding the expectation of $|W_i|$ does not suffice to show that our algorithm has a low work. Therefore, we need a proper extension of Lemma~\ref{lem_Wi} that exploits higher moments. Note that it works for arbitrary LP-type problems, i.e., also problems that are non-regular and/or non-degenerate.

\begin{lemma} \label{lem_hV}
Let $(H,f)$ be an LP-type problem of dimension $d$ and let $\mu$ be any multiplicity function. For any $k \ge 1$ and any $1 \le r < m/2-k$, where $m=|H(\mu)|$, it holds for $W_R=\{h \in H(\mu) \mid f(R)<f(R \cup \{h\}) \}$ for a random multiset $R$ of size $r$ from $H(\mu)$ that $\E[|W_R|^k] \le 2(k \cdot d \cdot (m-r)/(r+1))^k$.
\end{lemma}
\begin{proof}
By definition of the expected value it holds that
\[
  \E[|W_R|^k] = \frac{1}{{m \choose r}} \sum_{R \in {H(\mu) \choose r}} |W_R|^k
\]
For $R \in {H(\mu) \choose r}$ and $h \in H(\mu)$ let $X(R,h)$ be the indicator variable for the event that $f(R)<f(R \cup \{h\})$. Then
we have
\begin{eqnarray*}
 {m \choose r} \E[|W_R|^k] & = & \hspace*{-5mm} \sum_{R \in {H(\mu) \choose r}} |W_R|^k
   = \sum_{R \in {H(\mu) \choose r}} \left( \sum_{h \in H(\mu) - R} X(R,h) \right)^k \\
   & \stackrel{(1)}{\le} & \hspace*{-5mm} \sum_{R \in {H(\mu) \choose r}} \left( \sum_{h \in H(\mu) - R} X(R,h) \right. \\
   & & + 2^k \sum_{\{h_1,h_2\} \subseteq H(\mu) - R} \hspace*{-5mm} X(R,h_1) \cdot X(R,h_2) + \ldots \\
   & & + \left. k^k \sum_{\{h_1,\ldots,h_k \} \subseteq H(\mu) - R} \hspace*{-5mm} X(R,h_1) \cdot \ldots \cdot X(R,h_k) \right)
\end{eqnarray*} 
(1) holds because $X(R,h)^i = X(R,h)$ for any $i \ge 1$ and there are at most $i^k$ ways of assigning $k$ $X(R,h)$'s, one from each sum in $(\sum_{h \in H(\mu) - R} X(R,h))^k$, to the $i$ $X(R,h)$'s in some $X(R,h_1) \cdot \ldots \cdot X(R,h_i)$. Moreover, for any $k>1$,
\begin{eqnarray*}
  \lefteqn{ \sum_{R \in {H(\mu) \choose r}} \sum_{\{h_1,\ldots,h_k \} \subseteq H(\mu) - R} X(R,h_1) \cdot \ldots \cdot X(R,h_k) } \\
  & = & \sum_{Q \in {H(\mu) \choose r+k}} \sum_{\{h_1,\ldots,h_k \} \subseteq Q} X(Q-\{h_1,\ldots,h_k\},h_1) \cdot \ldots \\
  & & \hspace*{3cm} \cdot X(Q-\{h_1,\ldots,h_k\},h_k) \\
  & = & \sum_{Q \in {H(\mu) \choose r+k}} \sum_{\{h_1,\ldots,h_{k-1} \} \subseteq Q} \\
  & & \sum_{h_k \in Q-\{h_1,\ldots,h_{k-1}\}}  X(Q-\{h_1,\ldots,h_k\},h_1) \cdot \ldots \\ 
  & & \hspace*{3cm} \cdot X(Q-\{h_1,\ldots,h_k\},h_{k-1}) \\
  & & \hspace*{3cm} \cdot X((Q-\{h_1,\ldots,h_{k-1}\})-h_k,h_k) \\
  & \stackrel{(2)}{\le} & \sum_{Q \in {H(\mu) \choose r+k}} \sum_{\{h_1,\ldots,h_{k-1} \} \subseteq Q} \\
  & & \sum_{h_k \in B(Q-\{h_1,\ldots,h_{k-1}\})}  X((Q-h_k)-\{h_1,\ldots,h_{k-1}\},h_1) \cdot \ldots \\
  & & \hspace*{3cm} \cdot X((Q-h_k)-\{h_1,\ldots,h_{k-1}\},h_{k-1})
\end{eqnarray*}
\begin{eqnarray*}
  & \le & \sum_{Q \in {H(\mu) \choose r+k}} d \cdot \max_{h_k \in Q} \left(
    \sum_{\{h_1,\ldots,h_{k-1} \} \subseteq Q-h_k} \right. 
  X((Q-h_k)-\{h_1,\ldots,h_{k-1}\},h_1) \cdot \ldots \\
  & & \hspace*{2cm} \left.\phantom{\sum_Q} \cdot X((Q-h_k)-\{h_1,\ldots,h_{k-1}\},h_{k-1}) \right) \\
  & \le & ... \quad \le \sum_{Q \in {H(\mu) \choose r+k}} d^k
\end{eqnarray*}
where $B(S)$ is an optimal basis of $S$. (2) holds because $X((Q-\{h_1,\ldots,h_{k-1}\})-h_k,h_k)=0$ for every $h_k \not\in B(Q-\{h_1,\ldots,h_{k-1}\})$. The skipped calculations apply the same idea for $h_k$ to $h_{k-1},\ldots,h_2$. Hence, as long as $r+k < |H(\mu)|/2$,
\begin{eqnarray*}
 {m \choose r} \E[|W_R|^k] & = & \sum_{R \in {H(\mu) \choose r}} |W_R|^k \\
   & \le & \sum_{Q \in {H(\mu) \choose r+1}} d 
    + 2^k \hspace*{-2mm} \sum_{Q \in {H(\mu) \choose r+2}} \hspace*{-2mm} d^2 + \ldots
    + k^k \hspace*{-2mm} \sum_{Q \in {H(\mu) \choose r+k}} \hspace*{-2mm} d^k \\
   & \le & 2k^k \sum_{Q \in {H(\mu) \choose r+k}} d^k = 2(dk)^k {m \choose r+k}
\end{eqnarray*}
Resolving that to $\E[|W_R|^k]$ results in the lemma.
\end{proof}

Lemma~\ref{lem_hV} allows us to prove the following probability bound, which is
essentially best possible for constant $d$ by a lower bound in \cite{GW01}.

\begin{lemma} \label{lem_hWi}
Let $(H,f)$ be an LP-type problem of dimension $d$ and let $H(V)$ be any multiset of $H$ of size $m$. For any $\gamma \ge 1$ and $1 \le r < m/2 - \gamma$,
\[
  \Pr[|W_i| \ge 4\gamma \cdot \frac{d \cdot m}{n(r+1)} ] \le 1/2^\gamma
\]
\end{lemma}
\begin{proof}
From Lemma~\ref{lem_hV} and the Markov inequality it follows that, for any $c\ge 1$ and $k\ge 1$,
\[
  \Pr[|W_R|^k \ge c^k \cdot 2(k \cdot d \cdot (m-r)/(r+1))^k] \le 1/c^k
\]
and therefore,
\[
  \Pr[|W_R| \ge c \cdot (1+1/k)(k \cdot d \cdot (m-r)/(r+1))] \le 1/c^k
\]
Since, for every element $h \in H(V)$, the probability that $h \in H(v_i)$ is equal to $1/n$, it follows that
\[
  \Pr[|W_i| \ge 2c \cdot \frac{k \cdot d \cdot m}{n(r+1)} ] \le 1/c^k
\]
Setting $c=2$ and $k=\gamma$ results in the lemma.
\end{proof}

Therefore, w.h.p., $|W_i| = O(d \log n)$ for every $i$, so the maximum work needed for pushing some $W_i$ is $O(d \log n)$. Moreover, the size of $H(V)$ after $T$ rounds is at most $|H|+O(T d n \log n)$, w.h.p. On the other hand, we will show the following variant of Lemma~\ref{lem:llsuccess}.

\begin{lemma} \label{lem:hlsuccess}
Let $B$ be an arbitrary optimal basis of $H$. As long as no $v_i$ has satisfied $f(H(v_i))=f(H)$ so far, $\E[|\{ h \in H(V) \mid h \in B\}|] \ge 2^{T/d}$ after $T$ rounds of the High-Load Clarkson Algorithm.
\end{lemma}
\begin{proof}
Let $B=\{h_1(B),\ldots,h_b(B)\}$, $b \le d$, and let $p_{i,j}$ be the probability that $f(B_i)<f(B_i \cup \{h_j(B)\})$. If $f(B_i)<f(H)$, then there must exist an $h_j(B)$ with $f(B_i)<f(B_i \cup \{h_j(B)\})$, which implies that under the condition that $f(B_i)<f(H)$, $\sum_j p_{i,j} \ge 1$. Let $\rho_j$ be the expected number of duplicates created for some copy of $h_j(B)$. Since the $B_i$'s are sent to nodes chosen uniformly at random, $\rho_j = (1/n) \sum_i p_{i,j}$. Certainly, since $p_{i,j} \in [0,1]$ for all $i$, also $\rho_j \in [0,1]$. Moreover,
\[
  \sum_j \rho_j = \sum_j (1/n) \sum_i p_{i,j} 
  = (1/n) \sum_i \sum_j p_{i,j} \ge 1
\]
Hence, we can use the same arguments as in the proof of Lemma~\ref{lem:llsuccess}, with $p_j$ replaced by $\rho_j$ and without the term $(1+1/(2d))$ in the denominator since we do not perform filtering, to complete the proof.
\end{proof}

Thus, because $|H(V)| \le |H|+O(T d n \log n)$ after $T$ rounds, w.h.p., our algorithm must terminate within $O(d \log |H|) = O(d \log n)$ rounds, w.h.p.

For the termination detection, we can again use the algorithm proposed in Section~\ref{sec:termination}, which results in a maximum work of $O(\log n)$ per round.

\subsection{Accelerated High-Load Clarkson Algorithm}

If we are willing to spend more work, we can accelerate the High-Load Clarkson Algorithm. Suppose that in Algorithm~\ref{alg:DistClarkson2} node $v_i$ does not just push $B_i$ once but $C$ many times. Then the work for that goes up from $O(d)$ to $O(C \cdot d)$, and the maximum work for pushing out the elements of the $W_i$'s is now bounded by $O(C \cdot d \log n)$, w.h.p., which means that after $T$ rounds, $|H(V)|$ is now bounded by $|H|+O(T C \cdot d n \log n)$, w.h.p. Furthermore, we obtain the following result, which replaces Lemma~\ref{lem:hlsuccess}.

\begin{lemma} \label{lem:ahlsuccess}
Let $B$ be an arbitrary optimal basis of $H$. As long as no $v_i$ has satisfied $f(H(v_i))=f(H)$ so far, $\E[|\{ h \in H(V) \mid h \in B\}|] \ge (C+1)^{\lfloor T/d \rfloor}$ after $T$ rounds of the High-Load Clarkson Algorithm with parameter $C$.
\end{lemma}
\begin{proof}
Recall the definition of $\rho_j$ in the proof of Lemma~\ref{lem:hlsuccess}. It now holds that $\rho_j =(C/n) \sum_i p_{i,j}$, which implies that $\rho_j \in [0,C]$ for all $j$ and $\sum_j \rho_j \ge C$. Now, let $\rho_{j,t}$ be the expected number of duplicates created for some copy of $h_j(B)$ in round $t$, and fix any values of $\rho_{j,t}$ so that $\sum_j \rho_{j,t} \ge C$ and $\rho_j \in [0,C]$ for all $j$ and $t$. Let $\mu_{j,t}$ be the multiplicity of $h_j(B)$ at the end of round $t$. Then, for all $j$, $\mu_{j,0} \ge 1$, and
\[
  \E[\mu_{j,t+1}] \ge (1+\rho_{j,t}) \cdot \mu_{j,t}
\]
Hence, $\E[\mu_{j,T}] \ge \prod_{t=1}^T (1+\rho_{j,t})$. Suppose that $\sum_{t=1}^T \rho_{j,t} = M$. Since $\prod_{t=1}^T (1+\rho_{j,t})$ is a convex function (i.e., it attains its maximum when $\rho_{j,t}=\rho_{j,t'}$ for all $t,t'$ under the constraint that $\sum_{t=1}^T \rho_{j,t}$ is fixed, which can be seen from the fact that $((1+r)+\epsilon)((1+r)-\epsilon) = (1+r)^2 - \epsilon^2$), it gets lowest if as many of the $\rho_{j,t}$'s are as large as possible and the rest is 0. Thus, $\prod_{t=1}^T (1+\rho_{j,t}) \ge (C+1)^{\lfloor M/C \rfloor}$.

Since $\sum_{t=1}^T \sum_j p_{j,t} \ge C \cdot T$, there must be a $j^*$ with $\sum_{t=1}^T p_{j^*,t} \ge C \cdot T/d$.
Therefore, there must be a $j^*$ with $\E[\mu_{j^*,T}] \ge (C+1)^{\lfloor T/d \rfloor}$, which completes the proof.
\end{proof}

Setting $C=\log^{\epsilon} n$ for any constant $\epsilon>0$, it follows that our algorithm must terminate in $O((d/\epsilon) \log(|H|)/\log \log(n)) = O(d \log(n)/\log \log(n))$ rounds, w.h.p. This completes the proof of Theorem~\ref{th:main2}.

\section{LP-type Algorithm for the Hitting Set Problem}

Finally, we consider distributed algorithms for two NP-hard optimization problems, the hitting set problem and the set cover problem. Recall the definition of the hitting set problem $(X,{\mathcal S})$ and its formulation as an LP-type problem $(X,f)$ in Section~\ref{sec:new}.

Initially, the points in $X$ are randomly distributed among the nodes. Let
$X_0(v_i)$ be the set of elements in $X$ initially known to node $v_i$ to distinguish them from copies created later by the algorithm, and let $X_0 = \bigcup_i X_0(v_i)$.

In the following, let $s=|{\mathcal S}|$. We assume that $|X|=n$. Consider Algorithm~\ref{alg:DistClarkson4}. We assume the parameter $c$ to be a sufficiently large constant, and the parameter $r$ will be determined below. At any time during the algorithm, $X(v_i)$ denotes the multiset of elements in $X$ known to $v_i$ and $X(V)=\bigcup_i X(v_i)$. Let $m=|X(V)|$. We first show the following lemma.

\begin{algorithm}[th]
  \begin{algorithmic}[1]
    \Repeat
      \ForAll{nodes $v_i$ in parallel}
        \State choose a random multiset $R_i$ of size $r$ from $X(V)$
        \If{the selection of $R_i$ succeeds}
          \State ${\mathcal S}_i:=\{ S \in {\mathcal S} \mid S \cap R_i = \emptyset \}$
          \State $S:=$ random set in ${\mathcal S}_i$
          \State $W_i:=S \cap X(v_i)$
          \If{$|W_i| \le c \cdot d \log n$}
            \ForAll{$x \in W_i$} {\bf push}($x$)
            \EndFor
          \EndIf
        \EndIf
        \ForAll{$x$ received by $v_i$} add $x$ to $X(v_i)$
        \EndFor
        \ForAll{$h \in X(v_i) - X_0(v_i)$} 
          \State keep $x$ with probability $1/(1+1/(2d))$
        \EndFor
      \EndFor
    \Until{at least one $v_i$ satisfies ${\mathcal S}_i = \emptyset$}
  \end{algorithmic}
  \caption{Distributed Hitting Set Algorithm.}
  \label{alg:DistClarkson4}
\end{algorithm}

\begin{lemma}
Let $(X,f)$ be a hitting set problem with a minimal hitting set of size $d$ and let $X(V)$ be any multiset of $X$ of size $m$ containing $X$. Consider any node $v_i$. If $r \ge 6d \ln(12ds)$, then as long as the probability that $R_i$ is a hitting set is less then $1/2$, the expected size of $W_i$ is at most $m/(6dn)$.
\end{lemma}
\begin{proof}
For any set $S \in {\mathcal S}$ let $\mu(S)=|\{ x \in X(V) \mid x \in S\}|$. We say that $S$ is {\em small} if $\mu(S) \le m/(6d)$, and otherwise we say that $S$ is {\em large}. Let $e_-$ be the expected number of small sets $S$ not covered by $R_i$ and $e_+$ be the expected number of large sets $S$ not covered by $R_i$. If $e_-+e_+ < 1/2$, then it follows from the Markov inequality that the probability that $R_i$ is a hitting set is at least $1/2$ (which implies that w.h.p. there will be at least one $v_i$ whose $R_i$ is a hitting set, and the algorithm would terminate within at most $O(\log n)$ rounds). Thus, under the assumption that the 
probability of finding a hitting set is less then $1/2$, $e_-+e_+ \ge 1/2$. Let $p_-$ be the probability that the random set picked by $v_i$ is a small set and $p_+$ be the probability that the random set picked by $v_i$ is a large set. Under the assumption that $R_i$ is not a hitting set, it must hold that $p_-+p_+ = 1$ and the probabilities are proportional to $e_-$ and $e_+$, i.e., $p_- = e_-/(e_-+e_+)$ and $p_+=e_+/(e_-+e_+)$. For any small set $S$, the expected number of elements of $S$ stored in $v_i$ is at most $m/(6dn)$. Hence, $\E[|W_i|] \le m/(6dn)$ if the randomly selected set is small. If the randomly selected set is large, we can only give the trivial bound $\E[|W_i|] \le m/n$, but we can provide a bound on $p_+$.

Note that the probability that a large set $S$ is not covered is at most
$(1-1/(6d))^r \le e^{- r/(6d)} \le 1/(12ds)$. Since there are at most $s$ large sets, it follows that $e_+ \le 1/(12d)$. Since $e_-+e_+ \ge 1/2$, it therefore holds that $p_+ \le 1/(6d)$.Thus, overall,
\[
  \E[|W_i|] \le (1-p_+) \cdot m/(6dn) + p_+ \cdot m/n \le m/(6dn)
\]
\end{proof}

This allows us to prove the following lemma, whose proof follows the proof of Lemma~\ref{lem:Wgrowth}.

\begin{lemma}
As long as the probability of finding a hitting set is less than $1/2$, it holds for all $i$ that $|W_i|=O(m/n + \log n)$, w.h.p., and $\sum_{i=1}^n |W_i| \le m/(3d)$, w.h.p. 
\end{lemma}

With this lemma we can then show that $m$ will never be too large, so that the communication work of the nodes will never be too large. Its proof follows the proof of Lemma~\ref{lem:Hsize}.

\begin{lemma}
As long as the probability of sampling a hitting set is less than $1/2$,
the Hitting Algorithm satisfies $|X(V)|=O(n)$, w.h.p.
\end{lemma}

Once the probability of sampling a hitting set is at least $1/2$, the size of the $W_i$'s might get too large, but since a node $v_i$ only pushes out $W_i$ if $|W_i| \le c \cdot d \log n$, $|X(V)|$ can only grow by at most $O(dn \log n)$ in each round afterwards. However, once the probability of sampling a hitting set is at least $1/2$, at least one node will sample a hitting set, w.h.p., which means that after an additional $O(\log n)$ rounds (when using the termination checking in Section~\ref{sec:termination}), all nodes will terminate w.h.p. So $|X(V)|$ can never be larger than $O(n \log^2 n)$, w.h.p.

Moreover, the following lemma holds. Its proof is similar to the proof of Lemma~\ref{lem:llsuccess}.

\begin{lemma}
Let $T \in \N$ and $H$ be a minimal hitting set of $(X,{\mathcal S})$. After $T$ iterations, $\mu(H) \ge (2/\sqrt{e})^{T/d}$.
\end{lemma}

Hence, within at most $O(d \log n)$ rounds, their number of copies will exceed our bound on $|X(V)|$, which means the algorithm must terminate within $O(d \log n)$ rounds, w.h.p. When using the sampling and termination detection strategies of the Low-Load Clarkson algorithm, the maximum work in a round is bounded by $O(d \log(ds) + \log n)$, w.h.p. This completes the proof of Theorem~\ref{th:main3}.

\section{Experimental results}
While we have obtained the theoretical bound of $O(d\log n)$ rounds w.h.p. for our main two algorithms, the Low-Load Clarkson Algorithm (Algorithm~\ref{alg:DistClarkson}) and the High-Load Clarkson Algorithm  (Algorithm~\ref{alg:DistClarkson2}), we are also interested in their practical performance. In particular, we would like to estimate the constant factor hidden in our asymptotic bound. To achieve this, we will look at the specific LP-type problem of finding the minimum enclosing disk, i.e. the 2-dimensional version of the minimum enclosing ball problem mentioned in the introduction.

Note that the running time for the termination phase (Algorithm~\ref{alg:Termination}) of these algorithms is predictable and independent of the actual input, so we will measure the number of rounds until at least one node found the solution. We consider four different test cases \textbf{duo-disk}: 2 points lie on the solution disk with the remaining points uniformly distributed in the interior of the disk (Figure~\ref{subfig:duo_disk}), \textbf{triple-disk}: 3 points lie on the solution disk with the others uniformly distributed in the interior of the disk (Figure~\ref{subfig:triple_disk}), \textbf{triangle}: 3 points on a triangle with points uniformly distributed in the interior (Figure~\ref{subfig:triangle}), and \textbf{hull}: points at the vertices on a regular polygon that are slightly perturbed (Figure~\ref{subfig:hull}). For each test case, we take the average result of $10$ runs of our algorithms with $n$ nodes on $n$ data-points, where $n=2^i$ ranging over $i=1,\ldots 14$, (this is extended to $16$ for the 2-disk case for the low load algorithm), see Figures~\ref{fig:low_load} and~\ref{fig:high_load} for the results. 

For the low-load algorithm, note that the small test cases finish within one round, because there is a high probability that there is a node $v_i$ where $H(v_i)$ contains an optimal basis. For the duo-disk test case the number of rounds is $1.2 \log n$, while it is $1.7\log n$ for the other test cases. For the high-load algorithm, the runtime of the duo-disk test cases is around $0.9 \log n$, while it is $1.1 \log n$ for the other test cases. So 
the constants hidden in our asymptotic bounds are small. Note that the three test cases other than duo-disk behave similarly, while duo-disk runs a bit faster. The difference between the duo-disk case and the other test cases is the size of the optimal basis, which is $2$ for duo-disk and $3$ for the others. 
This suggests that the actual number of rounds depends on the size of the optimal basis for that particular problem and that other features of the problem do not influence the number of rounds much.

\begin{figure}
    \centering
    \begin{subfigure}{0.4\textwidth}
    	\includegraphics[width=\linewidth]{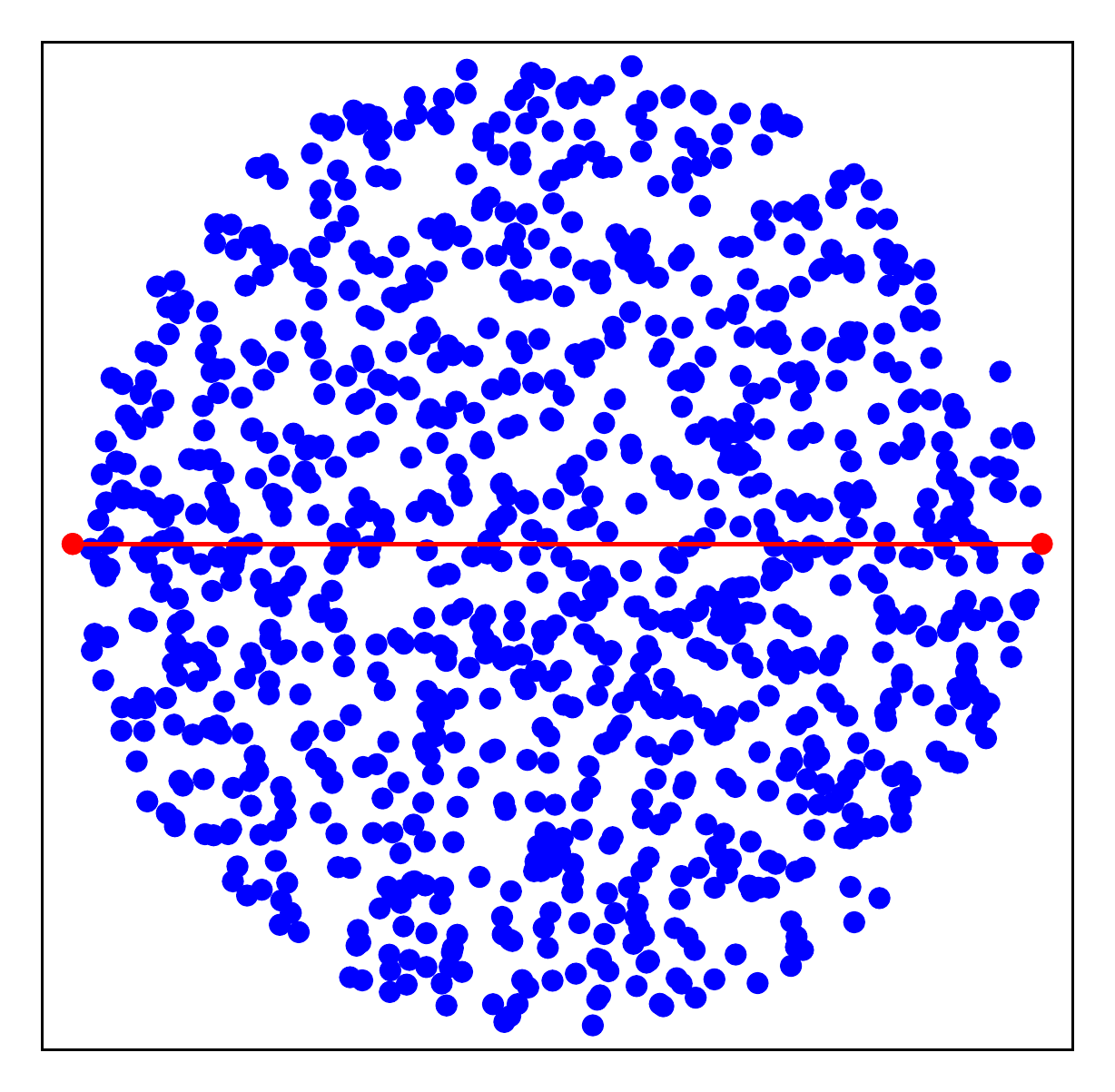}
	    \caption{\textbf{duo-disk}: The points are uniformly distributed over a disk defined by 2 points.\label{subfig:duo_disk}}
    \end{subfigure}
    \quad
    \begin{subfigure}{0.4\textwidth}
    	\includegraphics[width=\linewidth]{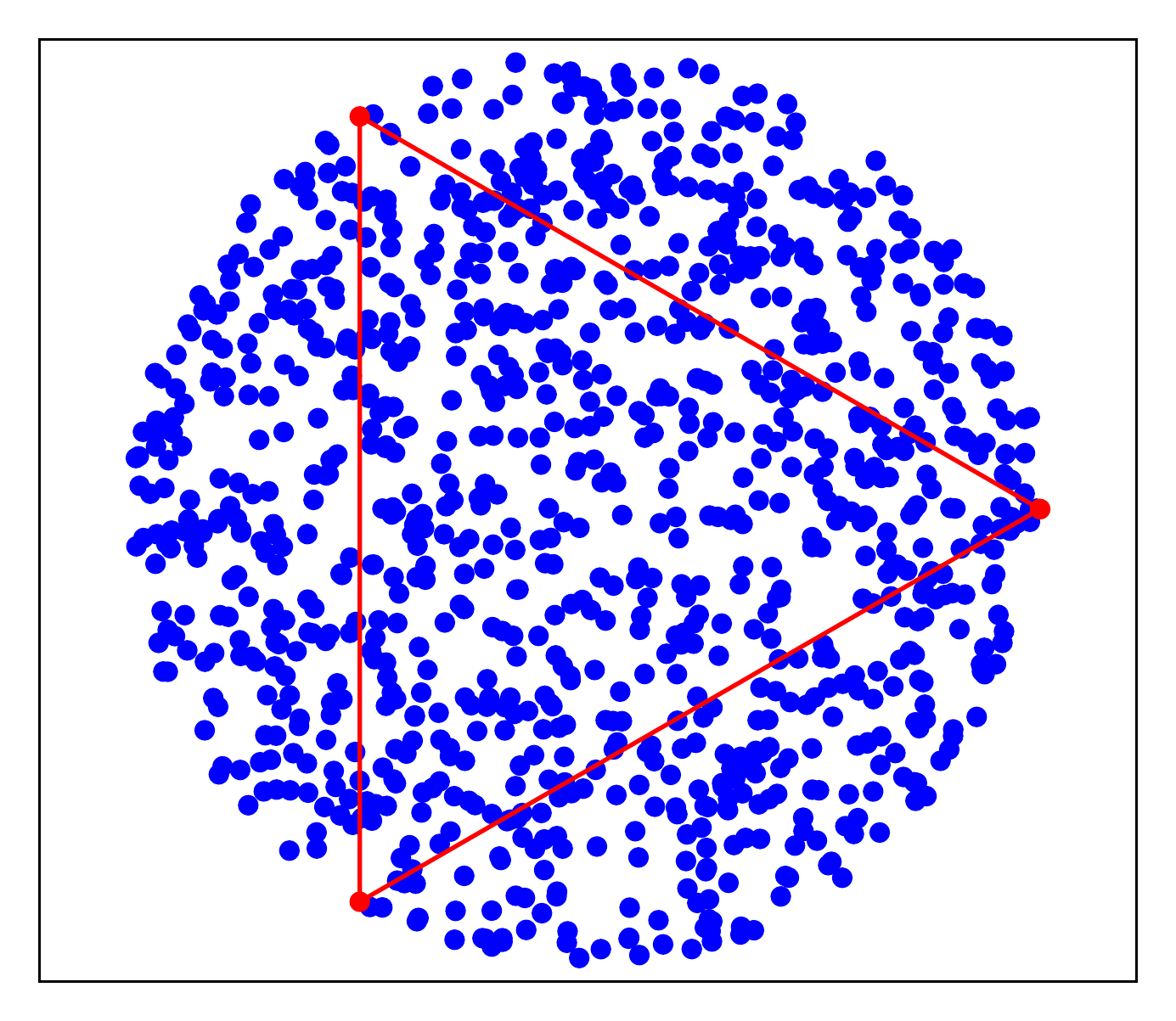}
	    \caption{\textbf{triple-disk}: The points are uniformly distributed over a disk defined by 3 points.\label{subfig:triple_disk}}
    \end{subfigure}
    \\
    \begin{subfigure}{0.4\textwidth}
	    \includegraphics[width=\linewidth]{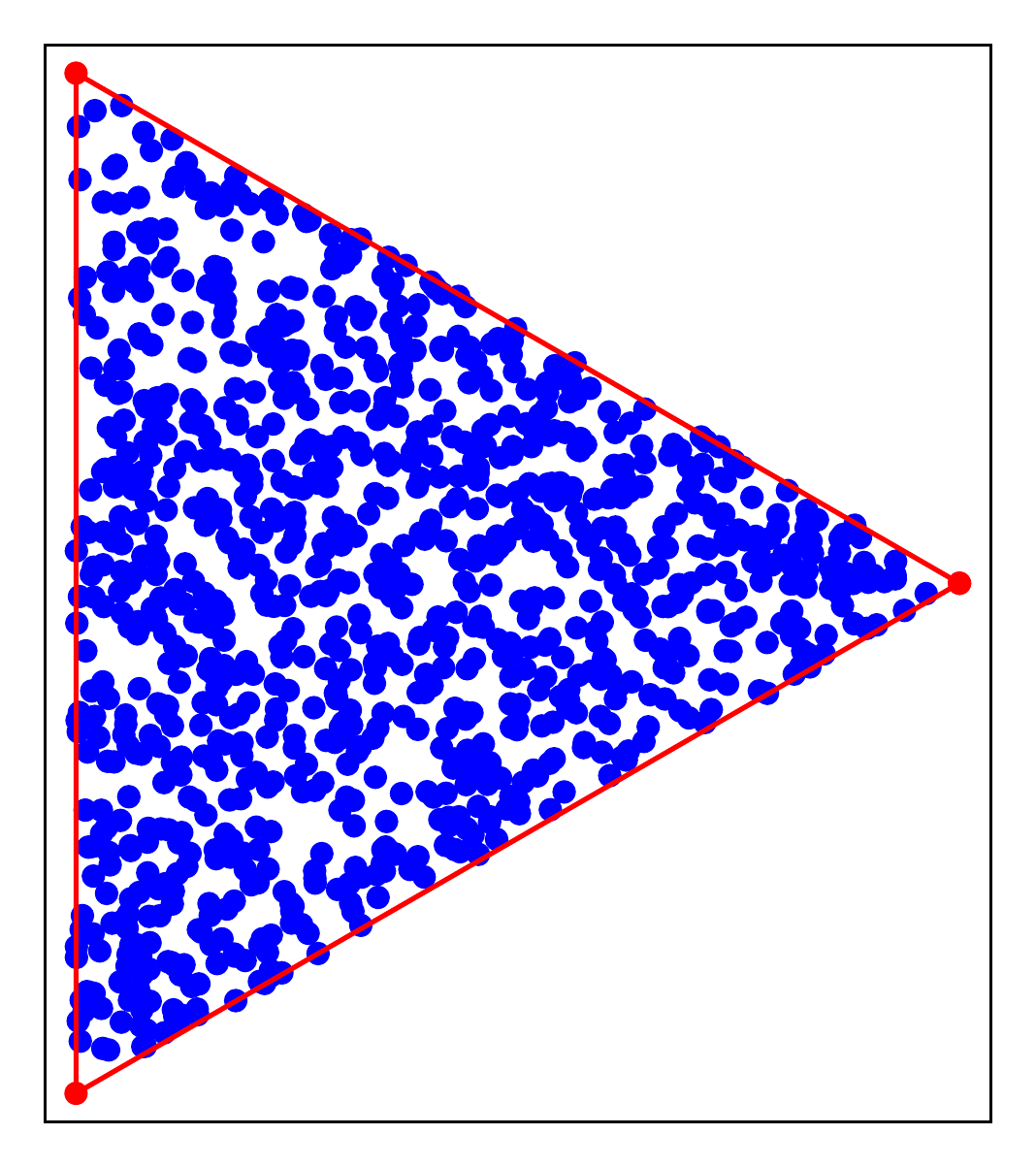}
	    \caption{\textbf{triangle}: The points are uniformly distributed over a triangle.\label{subfig:triangle}}
    \end{subfigure}
    \quad
    \begin{subfigure}{0.4\textwidth}
    	\includegraphics[width=\linewidth]{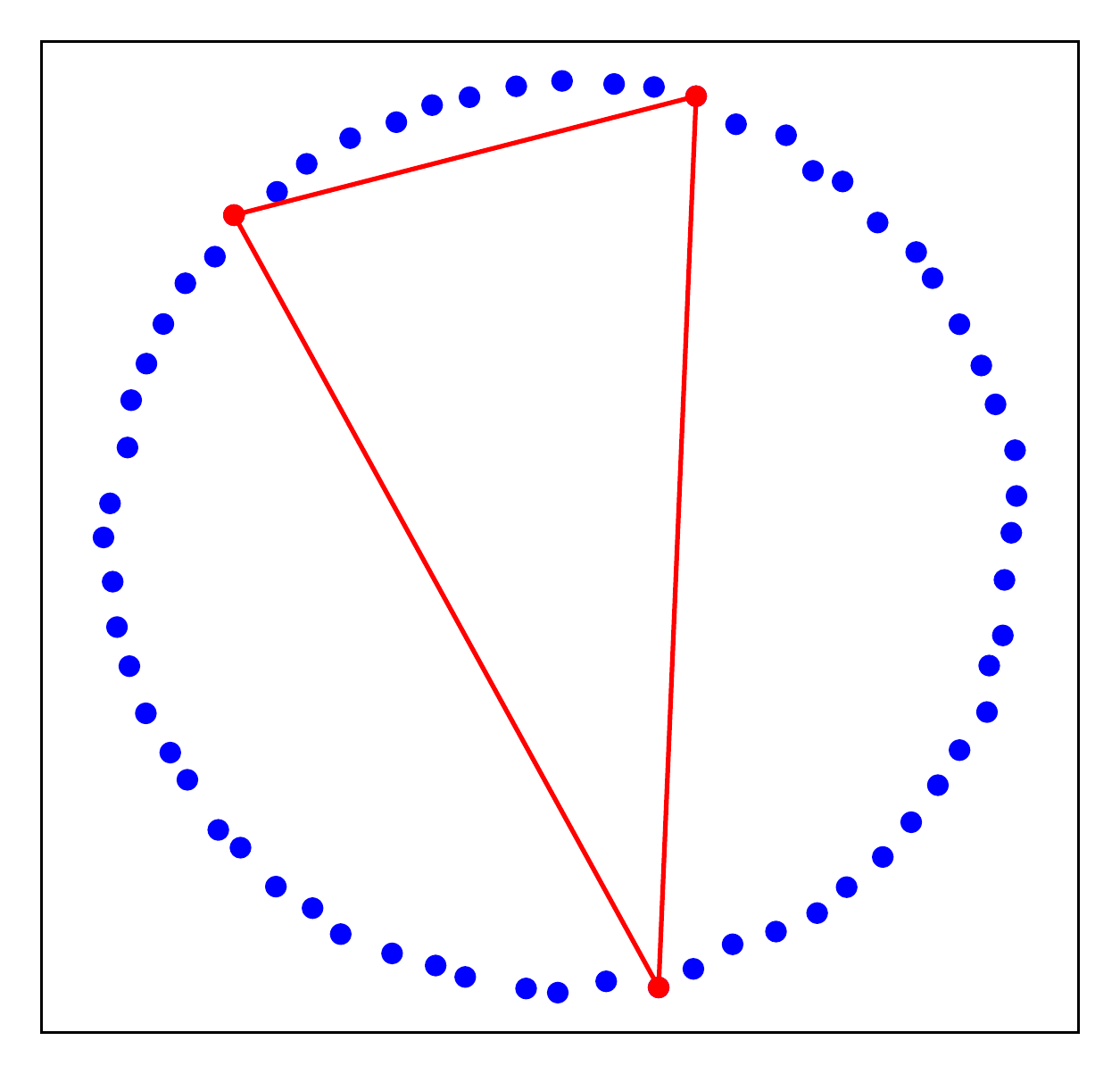}
	    \caption{\textbf{hull}: The points are perturbed vertices of a regular polygon.\label{subfig:hull}}
    \end{subfigure}
    \caption{\label{fig:datasets} The 4 types of data-sets of the minimum enclosing disk problem used in our experimental evaluation: duo-disk, triple-disk, triangle, and hull}
\end{figure}

\begin{figure}
    \centering
    \includegraphics[width=0.9\linewidth]{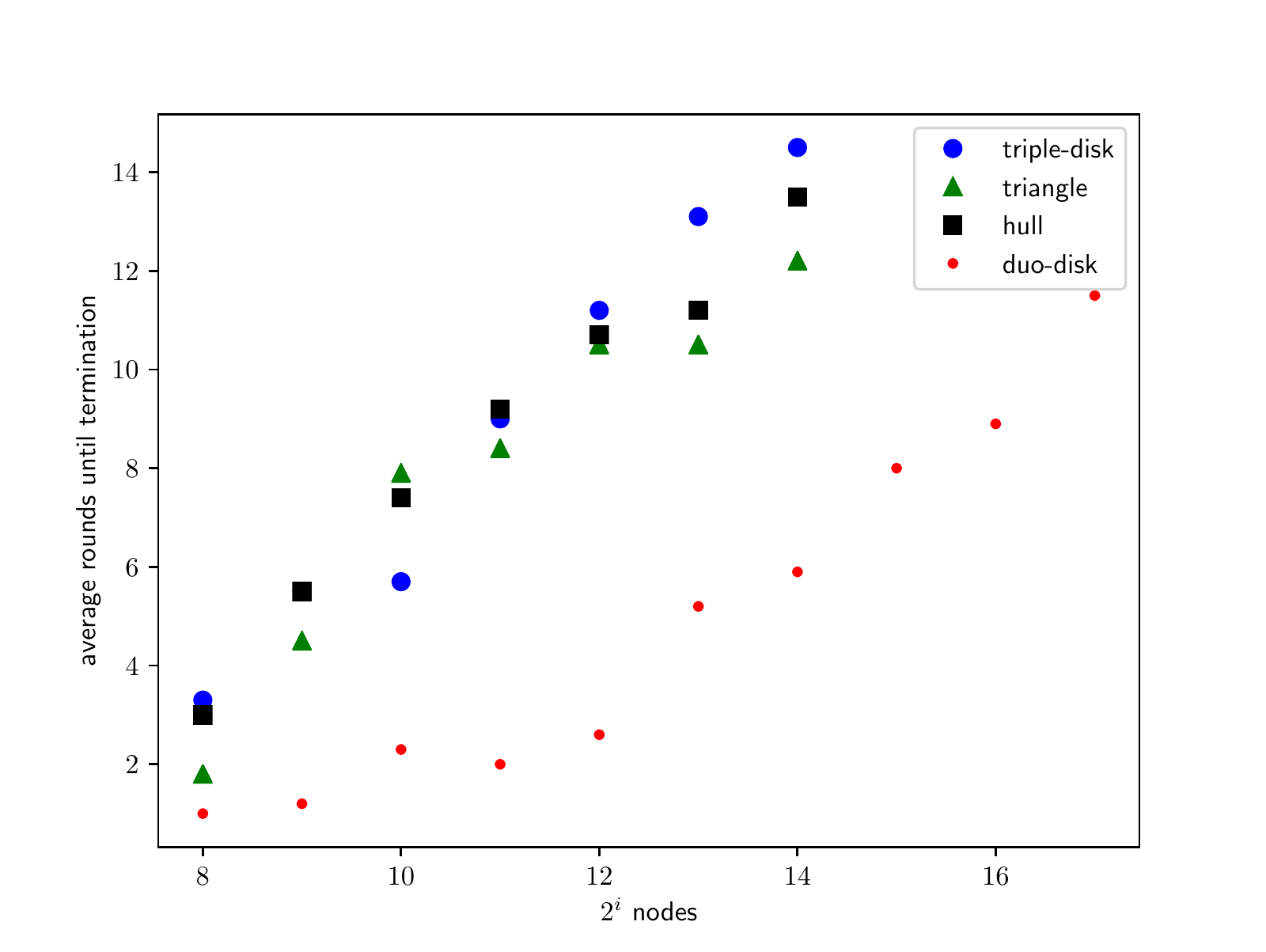}
    \caption{The average number of rounds until a node finds the minimum enclosing disk over $10$ runs of the Low-Load Clarkson Algorithm. Test instances of size $<2^8$ finish in one round.}
    \label{fig:low_load}
\end{figure}

\begin{figure}
    \centering
    \includegraphics[width=0.9\linewidth]{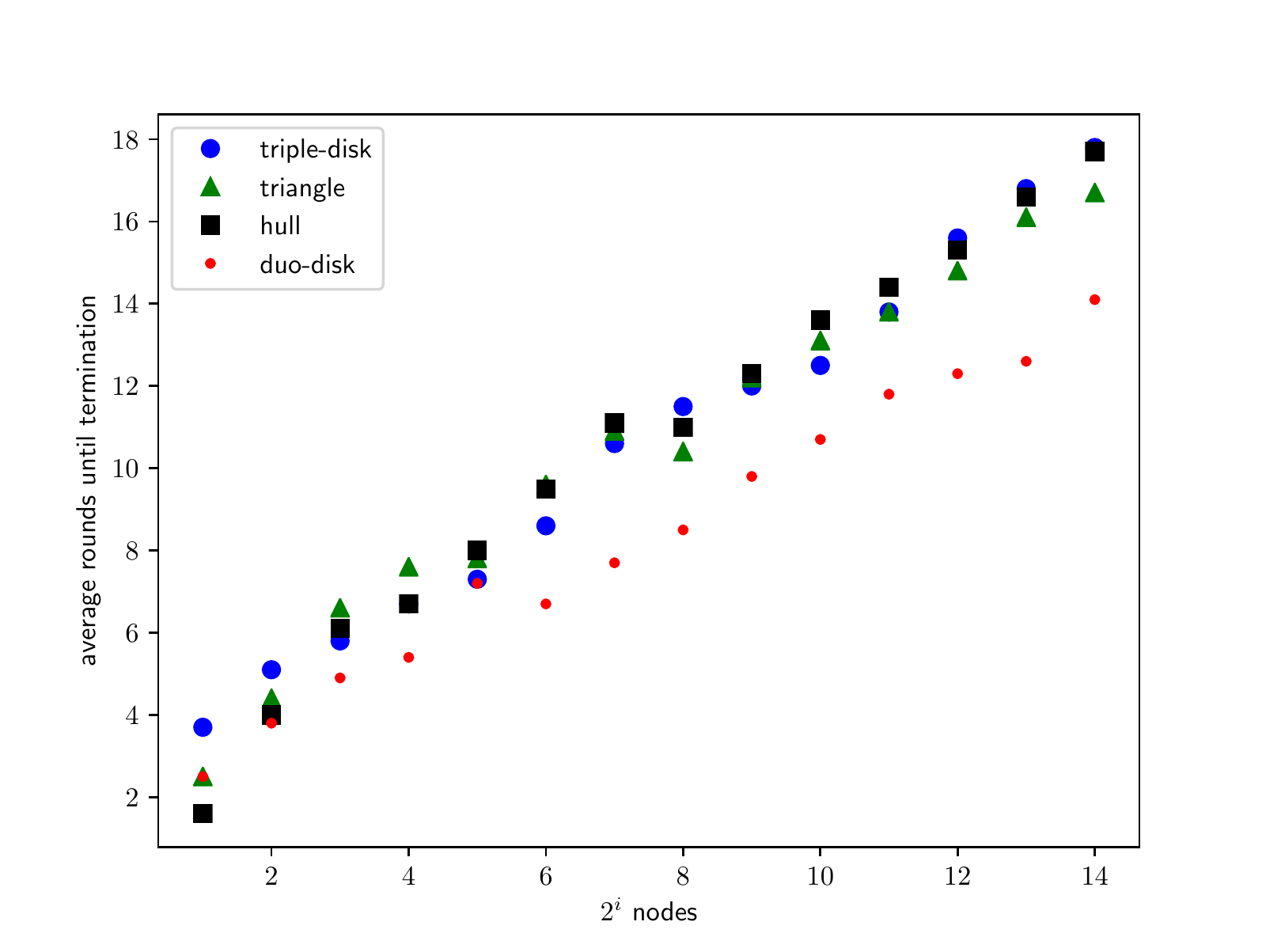}
    \caption{The average number of rounds until a node finds the minimum enclosing disk over $10$ runs of the High-Load Clarkson Algorithm.}
    \label{fig:high_load}
\end{figure}

\section{Conclusion}

In this paper we presented various efficient distributed algorithms for LP-type problems in the gossip model. Of course, it would be interesting to find out which other problems can be efficiently solved within Clarkson's framework, and whether some of our bounds can be improved.
\clearpage 
%
\bibliographystyle{plainurl}
\bibliography{literature}

\end{document}